\documentclass[a4paper, 1p, review]{elsarticle}

\usepackage{marginnote}
\usepackage[cmex10]{amsmath}
\usepackage[utf8]{inputenc}
\usepackage{amsmath}
\usepackage{color}
\usepackage[normalem]{ulem}
\usepackage{amsthm}

\newtheorem{theorem}{Theorem}

\newtheorem{proposition}{\bf Proposition}
\newtheorem{definition}{\bf Definition}

\newtheorem{remark}{Remark}

\newcommand{\old}[1]{{\color{red}{}}}

\begin{document}
\begin{frontmatter}

\title{Locating modifications in signed data for partial data integrity}

\author[ufsc]{Tha\'{i}s Bardini Idalino}
\address[ufsc]{Federal University of Santa Catarina, Florian\'{o}polis, Brazil}

\author[ottawa]{Lucia Moura}
\address[ottawa]{University of Ottawa, Ottawa, Canada}

\author[ufsc]{Ricardo Felipe Cust\'{o}dio}

\author[carleton]{Daniel Panario}
\address[carleton]{Carleton University, Ottawa, Canada}

\begin{abstract}
We consider the problem of detecting and locating modifications in signed data to
ensure partial data integrity. We assume that the data is  divided
into $n$ blocks (not necessarily of the same size) and that a threshold $d$ is given for the maximum
amount of modified blocks that the scheme can support. We propose efficient algorithms
for signature and verification steps which provide a reasonably compact signature size,
for controlled sizes of $d$ with respect to $n$. For instance, for fixed $d$ the standard signature size 
gets multiplied by a factor of $O(\log n)$, while allowing the identification of up to $d$ modified blocks.
Our scheme is based on nonadaptive combinatorial group testing and cover-free families.
\end{abstract}


\date{March 2015}
\journal{journal}

\end{frontmatter}

\section{Introduction}

Digital signature schemes can detect if modifications were done in a signed document, but do not offer information on where exactly those modifications occurred. 
In this context, even a single bit change would invalidate the whole document. In the present paper, we provide 
 a general Modification Location Signature Scheme, which determines which parts of the document were modified, thus ensuring
partial data integrity.
 
Partial data integrity is useful in several scenarios. First, we may need to ensure the integrity of specific parts of a document. For example, in fillable forms the owner may need to assure that the document is official, while some parts are expected to be modified.
Second,
in a data forensics investigation of a crime,
the investigator could have more clues on who 
is the attacker by knowing what exactly was modified  \cite{Goodrich}. 
Third, assuring that part of the data is intact can improve the efficiency of a computer system. For example, in a large database, the modification of some of its records would not invalidate
the whole database, avoiding total disruption of service. 

One can also see partial data integrity as a solution for guaranteeing privacy protection, where 
the extraction of selected portions of a signed document is to be shared with another party (content extraction signature \cite{Steinfeld}, redactable signature \cite{Johnson}).
 Our signature scheme capable of locating modifications can be used in this application by substituting the removed parts by ``blank" symbols. The original signature can be used to guarantee the integrity of the non-removed parts.

The Modification Location Signature Scheme (MLSS) proposed in this paper employs combinatorial
group testing to determine which blocks of a document contain modifications and which ones are intact.
This work is closely related to the work of Zaverucha and Stinson~\cite{zaverucha} who propose the use of group testing to identify modified documents in batch. However, while in~\cite{zaverucha}  group testing is used on the verifier's end to speed up the batch verification algorithm, 
in our approach it is used both at the signer's and verifier's end, which greatly improves the signature size over 
the trivial idea of treating each block of a document as an independent signed document. This trivial idea would require $n$ signatures for a document
divided into $n$ blocks, which would not be efficient, while for the cases of interest here we would have the size of a signature multiplied by a factor of $O(\log n)$ instead (see Theorem~\ref{analysis} and the discussion that follows it). While MLSS is applicable to any type of document (text, pictures, videos or a mix), the type of document may influence the way one divides it (see Section 3.4).

Definition of the problem and related work in digital signatures and group testing are presented in Section~\ref{sec:related}; the algorithms for the proposed Modification Location Signature Scheme and analysis are provided in Section \ref{sec:problem}; conclusions are given in Section \ref{sec:conclusion}.

\section{Definition of the problem and related work} \label{sec:related}
\subsection{Digital signatures}

Following a general definition~\cite{zaverucha}, a signature scheme is specified by algorithms ({\sc Gen, Sign, Verify}). {\sc Gen}$(k)$ receives a security parameter $k$ and outputs a pair of keys $(s_{k},p_{k})$, a secret key
used for signing and a public key used for verification, respectively.
{\sc Sign}$(s_{k},m)$ outputs a signature $\sigma$ on the message $m$ using the secret key $s_{k}$.  
{\sc Verify}$(p_{k},\sigma,m)$ outputs 1, using the public key $p_{k}$, if $\sigma$ is a valid signature of $m$, and 0 otherwise.

We propose a general digital signature scheme for signing a document divided into blocks providing, in the case
of modifications on the document after signing, the extra capability of locating which blocks have been modified.

\begin{definition}\label{def:MLSS} Modification Location Signature Scheme (MLSS):
Let $B=(B_1, \ldots,B_n)$ be a document divided into  $n$ blocks.
{\sc MLSS-Gen}($k$) receives a security parameter $k$ and outputs a pair of keys $(s_{k},p_{k})$.
{\sc MLSS-Sign}$(s_{k},B)$ outputs a signature $\sigma$ on  $B$ using the secret key $s_{k}$.
{\sc MLSS-Verify}$(p_{k},\sigma,B)$ outputs 1 if, using the public key $p_{k}$, $\sigma$ is a valid signature of $B$,
it outputs 0 if  $\sigma$  has been modified or is not authentic, 
and otherwise ($B$
has been modified) outputs extra information on the location of the modifications in $B$.
\end{definition}

In this paper, we present an approach, which is based on combinatorial group testing, to solve the challenge stated in Definition~\ref{def:MLSS}. 
Zaverucha and Stinson~\cite{zaverucha} observe that finding invalid signatures in a batch is a group testing problem, and propose the use of group testing methods to
improve the efficiency of the batch verification algorithm. In~\cite{zaverucha}, by exploring the best known group testing algorithms, they show how to
run $t$ signature verifications  to verify  a batch of $n$ signed documents, where $t$ is substantially smaller than $n$; in their case, adaptive and nonadaptive 
group testing can be used. We use a similar idea to improve the verification algorithm, but we also suggest applying group testing at the signer's
end in order to minimize the signature size.
We produce $t$ digests, each one involving a subset of the $n$ blocks
of the document ($t$ much smaller than $n$). This tuple of digests is signed and sent with the document, allowing the verifier to determine the blocks that were modified. 
This approach requires nonadaptive group testing, since the digests must be prepared at the signer's end independently of where modifications may occur. In this, case we need an upper bound $d$ on the number of modified blocks; this threshold value $d$
needs to be chosen carefully to keep control on the size $t$.

The presented method is not specific for asymmetric key encryption since one can choose any digital signature algorithm as {\sc Sign} and {\sc Verify}. In this paper, our presentation is based on public key digital signatures.

\subsection{Group testing}\label{sec:gt}
The purpose of group testing is to identify $d$ defective elements from a set of $n$ elements pooled into $t$ groups where $t<n$. The groups are tested, instead of all elements individually.
For a subset of elements (pool), if at least one of the elements is defective, we return the result for the test as a ``fail''; if no element is defective we return the result of the test as a ``pass''. In {\em adaptive} group testing, the results of the previous tests are used to
determine subsequent tests; in {\em nonadaptive} group testing, all the tests are specified ahead of time, which allows them to be run in parallel (see the book by Du and Hwang~\cite{gtLivro}). In our method, we need to use nonadaptive group testing, since the organization of blocks into
groups must be done at the signer’s end, and thus before the modifications (“defects’’) are introduced.
In this case, the most effective way to detect up to $d$ defectives is to use a cover-free family (CFF). 

\noindent
\begin{definition}\label{def:CFF}
A  \emph{$d$-cover-free family}, denoted $d$-CFF$(t,n)$ is a $t\times n$ binary matrix $M$ with $n$ $\geq$ $d+1$,
such that for any set of column indexes $C$ with $|C|= d$ and column $c\not \in C$, the following property holds:
there exists a row $i$ satisfying $M_{i,c}=1$ and $M_{i,j}=0$ for all $j \in C$.
\end{definition}

We form the tests according to the rows of matrix $M$, i.e. for each $1\leq i \leq t$, test $i$ consists of 
exactly the items $j$ for which $M_{i,j}=1$. \old{As stated in the next proposition,}The properties of CFFs assure that if the number of defectives is at most $d$ then  it is enough to determine the non-defective items from the passing tests. Then, we can conclude that all other items are defective.

Given $d$ and $n$, we wish to find a $d$-CFF$(t,n)$ for the smallest possible $t$, 
which we call $t(d,n)$.
We mention a few useful explicit constructions found in the literature.
When $d=1$,  we can use Sperner theorem \cite{sperner}
to show that  the smallest number of tests possible is $t(1,n)=\min\{t : {t \choose \lfloor t/2 \rfloor}\geq n\}$. We observe that as $n \rightarrow \infty$, $t(1, n) \sim \log_2 n$. The top-left of Fig.~\ref{fig:ourcheme} gives an example with $n=6$ and $t=4$. For arbitrary $d$ and $n$, we consider the constructions of Porat and Rothschild~\cite{PR} (\old{PR, }with $t \leq (d+1)^2 \ln n$) and  Pastuszak et al. \cite{PPS} (with $t\leq (d+1) \sqrt{n}$). Some of these constructions are surveyed in~\cite{zaverucha}. We note that for specific small $d$ one can find more efficient  constructions than the general ones
listed here; for example, for $d=2$ a smaller $t$ can be achieved (see~\cite{gtLivro}, Section 7.5). This more specific analysis is out of the scope of this paper.

\section{Location of modifications in signed documents}\label{sec:problem}
\subsection{Signature generation and verification algorithms} 
In our scheme, both signer and verifier use the same $t\times n$ matrix for a $d$-cover-free family via a call to a deterministic function {\sc MLSS-CFF}$(d, n)$, which can use constructions presented in Proposition~\ref{tab:values}, given later in this section.
We add a parameter $d$ in our algorithm, built on a given (traditional) signature scheme {\sc (Gen, Sign, Verify)}.
Algorithm {\sc MLSS-Gen}($k$) consists of a simple call to {\sc Gen}($k$).

Considering a document divided into blocks $B=(B_1,B_2,\ldots, B_n)$, algorithm {\sc MLSS-Sign} works as follows. Let $h_1,h_2,\ldots, h_n$ be the result of a public hashing algorithm $h(\cdot)$ applied on blocks $B_1,B_2,\ldots, B_n$, respectively. The tests, given by each row $i$ of the matrix, indicate which 
block hashes $h_j$ of the document are to be concatenated to form a test $T_i$ which is a digest of these concatenated hashes. Another digest $h^* = h(B)$ is calculated from $B$ yielding $T = (T_1,T_2,\ldots, T_t, h^{*})$. The signature of $B$ is given by $\sigma=(T, \sigma')$, where $\sigma'=\ ${\sc Sign}$(s_{k},T)$. We note that $T$ needs to be part of the signature since otherwise we would not be able to identify the modified blocks, as we present next.

At the other end, the {\sc MLSS-Verify} algorithm verifies if $\sigma$ is a valid signature by verifying $T$ with $\sigma'$. If so, then it compares $h^{*}$ with the hash of the received document $B'$. If they match, the document was not modified; otherwise, the algorithm locates the modified blocks as follows. Using the same method as the sender, 
it  computes $(T'_1,T'_2,\ldots,T'_t)$ from $B'$.
The set of indexes $i$ where $T'_i\not=T_i$
indicates which tests have failed and using group testing,
it  deduces exactly which blocks $B_j$ have been modified. This process is depicted in Fig.~\ref{fig:ourcheme}. Algorithm {\sc MLSS-Verify} also allows a faster verification that does not locate the modified blocks, much as a standard verification algorithm; this option is applied by setting the boolean location parameter $lc$ to false. The signature and verification algorithms are given next.

\begin{figure} [h]
\begin{center}
   \includegraphics[width=0.9\textwidth]{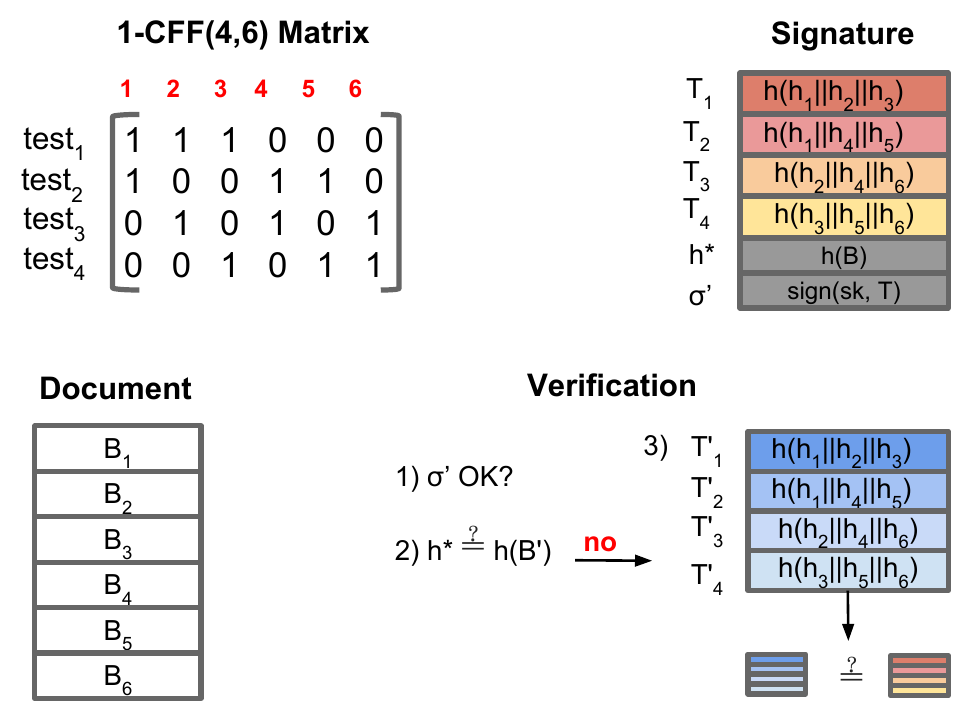}
   \caption{{\sc MLSS} signature and verification scheme.}
    \label{fig:ourcheme}
\end{center}
\end{figure}

\vspace{2mm}
\noindent
{\sc MLSS-Sign}$(s_{k},B,d)$
\\\textbf{Input:}  secret key $s_{k}$, document $B=(B_1,B_2,\ldots,B_n)$,  modification  threshold $d$.
\\\textbf{Output:} a signature $\sigma$.
\begin{enumerate}
\item Use $n$ and $d$ to determine $t$ and the $d$-CFF$(t,n)$ matrix $M$ to be used: \\$M =$ {\sc MLSS-CFF}$(d, n)$.
\item Let $h_j \leftarrow h(B_j)$, $1 \leq j \leq n$. Use $(h_1,h_2,\ldots, h_n)$ and $M$ to compute $T_1, \ldots , T_t$, as follows: for each row $1\leq i\leq t$ compute $c_i$, which is  the concatenation of the hashes $h_j$ for $j$ such that $M_{i,j}=1$, and let $T_i=h(c_i)$. Calculate $h^* = h(B)$ and set $T=(T_1,T_2,\ldots, T_t, h^*)$.
\item Compute $\sigma'=\ ${\sc Sign}$(s_{k},T)$ and output  $\sigma=(T, \sigma')$.
\end{enumerate}

The verification algorithm has three possible outcomes: signature has been modified (output 0); signature and document were not modified (output 1); signature was not modified and document has been modified (output I as the set of modified block indices, if $lc$ = true; output 2, otherwise).

\vspace{2mm}
\noindent
{\sc MLSS-Verify}$(p_{k},\sigma,B',d, lc)$ 
\\\textbf{Input:} public key $p_{k}$, signature $\sigma$, document $B'=(B'_1,B'_2,\ldots,B'_n)$, modification  threshold $d$, boolean location parameter $lc$.
\\\textbf{Output:}  0, 1, 2 or I, corresponding to the outcomes above explained.
\begin{enumerate}
\item Verify $\sigma$: Let $\sigma=(T,\sigma')$ and $T=(T_1,T_2,\ldots, T_t, h^*)$. If {\sc Verify}$(p_{k},T,\sigma')=0$ then output 0 and exit.
\item Compute $h^{**} \leftarrow h(B')$. If $h^* = h^{**}$ then output 1 and exit.
\item If $lc$ = false, then output 2 and exit. 
\item Use $n$ and $d$ to determine the $d$-CFF$(t,n)$ matrix $M$ to be used: \newline $M =$ {\sc MLSS-CFF}$(d, n)$.
\item Use the same process as Step 2 of {\sc MLSS-Sign} to compute $T'=(T'_1,\ldots, T'_t)$ from $B'$ using $M$ and $h(\cdot)$.
\item Compare $T$ and $T'$ and record the discrepancies (failing tests): \newline $F=\{i \in \{1, \ldots, t\}: T_i \not = T'_i\}$.
\old{ \item If $F=\emptyset$ then output 1 and exit.} 
\item Use group testing to determine the modified blocks: \label{thisstep}\\
	\hspace{4mm} Initialize $I=\{1,\ldots,n\}$;\\
	\hspace{4mm} for each $i\not\in  F, 1 \leq i \leq t,$ do\\
	\hspace*{8mm} for each $j \in I$ such that $M_{i,j}=1$ do $I \leftarrow I\setminus\{j\}$; \\
	\hspace{4mm} output $I$.
\end{enumerate}

\begin{remark}
If the number of modified blocks is larger than $d$, then the algorithm outputs a set $I$, with $|I| > d$, that contains all the modified blocks and possibly some more. However, any block that is not in $I$ was not modified.
\end{remark}


\subsection{Correctness and complexity of the algorithms}

\begin{proposition}
Consider a document and its signature generated by {\sc MLSS-Sign}. Then, algorithm {\sc MLSS-Verify} correctly verifies the signature according to the three possible outcomes. In particular, if the signature is valid and there are up to $d$ modifications, then $I$ is precisely the set of indices of these modified blocks.
\end{proposition}

\begin{proof}
Steps 1 and 2 of {\sc MLSS-Verify} identify the case of invalid signature and the case of valid signature with unmodified document, respectively. The next steps deal with the case of valid signature and modified document. Depending on the boolean parameter $lc$, we exit at Step 3 or move on to locate the modified blocks.
Steps 4 to 6 perform the hash concatenations dictated by matrix $M$ to reproduce the creation of tests $(T'_1,\ldots, T'_t)$ from $B'$.
The correctness of Step~\ref{thisstep} of {\sc MLSS-Verify} follows directly from the properties of a cover-free family,
and the assumption that the hash function used has the desired property (no collisions).
A matching test $T_i =T'_i$ guarantees that all the blocks that are concatenated to produce $T_i$ have not been modified. If the total number of modified blocks is at most $d$, then every unmodified block $B_j$ is part of some matching test $T_i$, and the remaining blocks are precisely the modified blocks. 
\end{proof}

We now analyze the algorithms and compare them with traditional signature schemes. Denote by $comp(x)$ the cost of comparing $x$ bits. Let $b$ be the size of $B$ in bits, $w$ be the number of 1's on the CFF matrix $M$, and $t$ be the number
of rows in $M$. Denote by $cost_{CFF}(d,n)$ the cost of computing function {\sc MLSS-CFF$(d,n)$}.

\begin{theorem}\label{analysis}
Consider {\sc MLSS-Sign} and {\sc MLSS-Verify} algorithms. Assume that the cost (running time) of computing the hash function $h$, denoted by $cost_{h}$, is a linear function on the input size,  
and let $h_{out}$ denote the number of bits of the output of $h$.  Assume algorithm {\sc Sign} {\sc (Verify)} first applies function $h$ on its input message and then applies a signature method (verification method) with cost denoted by $cost_{\mbox{sign}}(h_{out})$ $(cost_{\mbox{verify}}(h_{out}))$.
Then,
\begin{enumerate}
\item The size of the signature $\sigma$ produced by {\sc MLSS-Sign} is $(t+1) h_{out}+|\sigma'|$ while the size produced by a traditional
 signature method is $|\sigma'|$.
\item The running time of {\sc MLSS-Sign} is $cost_{h}(2b+(w+t+1) h_{out}) + cost_{\mbox{sign}}(h_{out})+cost_{CFF}(d,n)$, while the running time of {\sc Sign} is $cost_{h}(b) + cost_{\mbox{sign}}(h_{out})$. 
\item The running time of {\sc MLSS-Verify} when the signature is invalid (output 0), or the document has not been modified (output 1) is
$cost_h(b+(t+1) h_{out})+cost_{\mbox{verify}}(h_{out})+comp(h_{out})$ (this is also the cost when $lc$ = false); 
when the document has been modified but the signature has not and $lc$ = true, the running time of {\sc MLSS-Verify} is
$cost_{h}(2b+(w+t+1) h_{out}) + cost_{\mbox{verify}}(h_{out})+cost_{CFF}(d,n) + $ \\$comp((t+1)h_{out}) + c.w$, where $c$ is a constant. 
The running time of {\sc Verify} is $cost_{h}(b) + cost_{\mbox{verify}}(h_{out})$.

\end{enumerate}
\end{theorem}

\begin{proof}
The size of the signature $\sigma$ comes directly from its form, composed by  $t+1$ hashes $T$ and the signature $\sigma'$.
In Step 1, {\sc MLSS-Sign} computes the matrix $M$; in Step 2 it computes $n$ hashes with total input size $b$, followed by $t$ hashes of total input size $w \cdot h_{out}$
plus a hash of the entire document with input size $b$; and in Step 3 we have assumed that {\sc Sign} applies a 
hash on $T$, which has size $(t+1)h_{out}$. Hence, the linearity of  $cost_h$ yields $cost_{h}(2b+(w+t+1) h_{out})$ for hash computations
plus the cost of signing a message of size $h_{out}$ and computing matrix $M$.
If we apply algorithm {\sc Sign} directly to the message $B$ we have $cost_{h}(b) + cost_{\mbox{sign}}(h_{out})$, instead.\\
In {\sc MLSS-Verify}, Step 1 yields the cost of {\sc Verify} of 
an input of size $(t+1)h_{out}$, and Step 2 uses a hash of the whole document plus a comparison of $h_{out}$ bits, giving $cost_h(b+(t+1)\cdot h_{out})+cost_{\mbox{verify}}(h_{out})+comp(h_{out})$.
In the case of blocks that have been modified with a valid signature and $lc$ = true, the cost incurred by Steps 4 to 6 is $cost_h(b+ w\cdot h_{out}) + comp(t\cdot h_{out})+cost_{CFF}(d,n)$. 
Step 7 can be done in time linear with $w$.
\end{proof}





We now discuss practical implications of Theorem~\ref{analysis}.
The most significant cost in digital signature algorithms is related to the use of cryptographic functions. If we compare our algorithms with the traditional ones, we note that this cost remains the same, while most of our extra cost comes from additional hash computations.
In Section 3.3, we give a detailed comparison of these algorithms based on a standard digital signature method. As an illustration, for a document divided into $n = 256$ blocks with 1024 bytes per block and $d=10$, the running time of {\sc MLSS-Sign} (the calculation of $cost_h$ and $cost_{sign}$) is 4.8561 ms, while a traditional {\sc Sign} costs 2.8804 ms. If the document is not modified or the signature is invalid or $lc$ = false, the running time of {\sc MLSS-Verify} is 1.5246 ms while {\sc Verify} is 1.4934 ms. If there are modifications but the signature is valid and \\$lc$ = true, we locate $d = 10$ modified blocks using a total running time of 3.4691 ms. More experiments are provided in Section 3.3.

We note that the signature $\sigma$ generated by {\sc MLSS-Sign} has an additional size of $(t+1) h_{out}$ bits. In order to keep $|\sigma|$ as small as possible we need to minimize $t$, the number of rows in the CFF matrix. To obtain a small enough $t$, we consider constructions of Sperner~\cite{sperner} (S1), Porat and Rothschild~\cite{PR} (PR), Pastuszak et al.~\cite{PPS} (PPS), as well as using identity matrix $I_n$; we use the relatively better $t$ given $d$ and $n$ to choose one construction among these, as given in Proposition~\ref{CFF-calc}. A small $t$ is also important to reduce the extra computation costs for signing and verifying; moreover, given $t$
it is desirable to choose a CFF matrix with the smallest $w$.

\begin{proposition}~\label{CFF-calc}
Let $d,n$ be integers, and let $t$ be the number of rows and $w$ be the number of ones of the 
CFF matrix given by function {\rm MLSS-CFF}$(d,n)$.
Then the table below gives the values of $t$ and $w$, for each the four constructions specified.
\begin{center}
\label{tab:values}
\begin{tabular}{|l|c|c|} \hline
ranges of $d,n$ & t & w\\ \hline
$d=1$, for all $n$: use {\rm S1} & $\sim \log_2 n$ & $\sim \lfloor{\frac{\log_2 n}{2}}\rfloor n$\\ \hline
$d\in[2,\frac{\sqrt{n}}{\ln n})$: use {\rm PR} & $(d+1)^2\ln n$ & $\frac{(d+1)}{2} n \ln n$ \\ \hline
$d\in[\frac{\sqrt{n}}{\ln n},\sqrt{n}-1)$: use {\rm PPS} & $(d+1) \sqrt{n}$ & $n(d+1)$\\ \hline
$d\geq \sqrt{n}-1$: use ${I_n}$ & $n$ & $n$\\ \hline
\end{tabular}
\end{center}
\end{proposition}

\begin{proof}
For $d=1$, Sperner theorem gives $t =\min\{s : {s \choose \lfloor s/2 \rfloor}\geq n\}$ and $t \rightarrow \log_2 n$ as $n \rightarrow \infty$. Each column of the matrix has $\lfloor t/2 \rfloor$ ones and the matrix has $n$ columns. Using the proof of Theorem 1 in Section 4 of~\cite{PR} we obtain $w$ and $t$ for construction PR. The values of $w$ and $t$ in construction PPS come directly from Definition 3 in~\cite{PPS}.
\end{proof}

\subsection{Experimental results}


Here we provide some experiments using a standard digital signature (SHA256 with RSA 2048 bits) and \emph{openssl}. We quantify the cost of computing $cost_h$, while ignoring
other less relevant linear costs ($comp(.)$, $c\cdot w$) and the cost of CFF matrix computation, which could be preprocessed for specific applications. Indeed, we experimentally
verify\footnote{in an iMac 2.7 GHz Intel Core i5 with 6 MB on-chip L3 cache.} that hash computations behave in {\em openssl} as a linear function of the input size (which is linear in $t$ and $w$),  approximately as $5.52 \times 10^{-9} x + 3.819 \times 10^{-7}$, agreeing with the assumptions of Theorem~\ref{analysis}.

In the following tables we give running times obtained experimentally for a set of chosen values of the document size $b$, the number of blocks $n$ and the number of tests $t$, where the costs are given in milliseconds and the document size $b$ in bytes. In the tables below, ``{\sc MLSS-S}'' refers to the running time of {\sc MLSS-Sign}, ``{\sc MLSS-V.1}'' refers to the running time of {\sc MLSS-Verify} when there is no modifications or the signature is invalid or $lc$ = false, while ``{\sc MLSS-V.2}'' gives {\sc MLSS-Verify} running time when modifications occurred and $lc$ = true. 
Different tables use block sizes of 1024 and 8192 bytes. 

\begin{table*}[ht!]
\small
\centering
\caption{Comparison between methods with blocks of size 8192 bytes and $d = 1$.}
\begin{tabular}{c|c|c|c|c|c|c|c|c}
\hline $b$ & $n$ & $t$ & $w$ & {\sc Sign} & {\sc MLSS-S} & {\sc Verify} & {\sc MLSS-V.1} & {\sc MLSS-V.2}\\ 
\hline 16,384 & 2 & 2 & 2 & 1.5238 & 1.6151 & 0.1368 & 0.1373 & 0.2281\\
\hline 65,536 & 8 & 5 & 20 & 1.7951 & 2.1614 & 0.4081 & 0.4092 & 0.7744\\
\hline 262,144 & 32 & 7 & 112 & 2.8804 & 4.3486 & 1.4934 & 1.4948 & 2.9616\\
\hline 1,048,576 & 128 & 10 & 640 & 7.2215 & 13.1246 & 5.8345 & 5.8364 & 11.7376\\
\hline
\end{tabular}

\vspace*{2mm}

\small
\centering
\caption{Comparison between methods with blocks of size 1024 bytes and $d = 1$.}
\begin{tabular}{c|c|c|c|c|c|c|c|c}
\hline $b$ & $n$ & $t$ & $w$ & {\sc Sign} & {\sc MLSS-S} & {\sc Verify} & {\sc MLSS-V.1} & {\sc MLSS-V.2}\\ 
\hline 16,384 & 16 & 6 & 48 & 1.5238 & 1.6239  & 0.1368 & 0.1380 & 0.2369\\
\hline 65,536 & 64 & 8 & 256 & 1.7951 & 2.2037 & 0.4081 & 0.4097 & 0.8167\\
\hline 262,144 & 256 & 11 & 1408 & 2.8804 & 4.5782 & 1.4934 & 1.4955 & 3.1912\\
\hline 1,048,576 & 1024 & 13 & 6656 & 7.2215 & 14.1878 & 5.8345 & 5.8369 & 12.8008\\
\hline
\end{tabular}

\vspace*{2mm}
\small
\centering
\caption{Comparison between methods with blocks of size 8192 bytes and $d = 2$.}
\begin{tabular}{c|c|c|c|c|c|c|c|c}
\hline $b$ &       $n$ & $t$ & $w$           & {\sc Sign}& {\sc MLSS-S} & {\sc Verify} & {\sc MLSS-V.1} & {\sc MLSS-V.2}\\ 
\hline 16,384     & 2 & 2 & 2                  & 1.5238       & 1.6151             & 0.1368         & 0.1373          & 0.2281\\
\hline 65,536     & 8 & 8 & 8                  & 1.7951       & 2.1599            & 0.4081         & 0.4097 & 0.7729\\
\hline 262,144   & 32 & 16 & 96            & 2.8804       & 4.3475            & 1.4934         & 1.4965 & 2.9605\\
\hline 1,048,576 & 128 & 33 & 384          & 7.2215       & 13.0836          & 5.8345         & 5.8406           & 11.6966\\
\hline
\end{tabular}

\vspace*{2mm}

\small
\centering
\caption{Comparison between methods with blocks of size 1024 bytes and $d = 2$.}
\begin{tabular}{c|c|c|c|c|c|c|c|c}
\hline $b$ & $n$ & $t$ & $w$             & {\sc Sign}& {\sc MLSS-S} & {\sc Verify} & {\sc MLSS-V.1} & {\sc MLSS-V.2}\\ 
\hline 16,384 & 16 & 12 & 48               & 1.5238     & 1.6250             & 0.1368         & 0.1391 & 0.2380\\
\hline 65,536 & 64 & 24 & 192             & 1.7951     & 2.1952             & 0.4081         & 0.4125                & 0.8082\\
\hline 262,144 & 256 & 48 & 768       & 2.8804     & 4.4717             & 1.4934         & 1.5020                & 3.0847\\
\hline 1,048,576 & 1024 & 62 & 10646 & 7.2215     & 14.9014           & 5.8345         & 5.8457                & 13.5144\\
\hline
\end{tabular}

\end{table*}

\vspace*{2mm}

\begin{table*}[ht!]
\small
\centering
\caption{Comparison between methods with blocks of size 8192 bytes and $d = 3$.}
\begin{tabular}{c|c|c|c|c|c|c|c|c}
\hline $b$ &       $n$ & $t$ & $w$           & {\sc Sign}& {\sc MLSS-S} & {\sc Verify} & {\sc MLSS-V.1} & {\sc MLSS-V.2}\\ 
\hline 16,384     & 2 & 2 & 2                & 1.5238       & 1.6151       & 0.1368         & 0.1373          & 0.2281\\
\hline 65,536     & 8 & 8 & 8                  & 1.7951       & 2.1599       & 0.4081         & 0.4097           & 0.7729\\
\hline 262,144   & 32 & 22 & 128              & 2.8804       & 4.3542       & 1.4934         & 1.4975           & 2.9672\\
\hline 1,048,576 & 128 & 45 & 512          & 7.2215       & 13.1082         & 5.8345         & 5.8426           & 11.7212\\
\hline
\end{tabular}
\vspace*{2mm}

\small
\centering
\caption{Comparison between methods with blocks of size 1024 bytes and $d = 3$.}
\begin{tabular}{c|c|c|c|c|c|c|c|c}
\hline $b$ & $n$ & $t$ & $w$              & {\sc Sign}& {\sc MLSS-S} & {\sc Verify} & {\sc MLSS-V.1} & {\sc MLSS-V.2}\\ 
\hline 16,384 & 16 & 16 & 16                & 1.5238     & 1.6200           & 0.1368         & 0.1398                & 0.2330\\
\hline 65,536 & 64 & 32 & 256              & 1.7951     & 2.2079           & 0.4081         & 0.4139 & 0.8209\\
\hline 262,144 & 256 & 64 & 1024        & 2.8804     & 4.5198             & 1.4934         & 1.5048   & 3.1328\\
\hline 1,048,576 & 1024 & 110 & 14195  & 7.2215     & 15.5369           & 5.8345         & 5.8542                & 14.1499\\
\hline
\end{tabular}


\vspace*{2mm}

\small
\centering
\caption{Comparison between methods with blocks of size 8192 bytes and $d = 10$.}
\begin{tabular}{c|c|c|c|c|c|c|c|c}
\hline $b$ &       $n$ & $t$ & $w$           & {\sc Sign}& {\sc MLSS-S} & {\sc Verify} & {\sc MLSS-V.1} & {\sc MLSS-V.2}\\ 
\hline 16,384     & 2 & 2 & 2                  & 1.5238       & 1.6151             & 0.1368         & 0.1373          & 0.2281\\
\hline 65,536     & 8 & 8 & 8                  & 1.7951       & 2.1599            & 0.4081         & 0.4097           & 0.7729\\
\hline 262,144   & 32 & 32 & 32              & 2.8804       & 4.3389            & 1.4934         & 1.4992 & 2.9519\\
\hline 1,048,576 & 128 & 124 & 1408        & 7.2215       & 13.2805          & 5.8345         & 5.8566           & 11.8935\\
\hline
\end{tabular}

\vspace*{2mm}
\small
\centering
\caption{Comparison between methods with blocks of size 1024 bytes and $d = 10$.}
\begin{tabular}{c|c|c|c|c|c|c|c|c}
\hline $b$ & $n$ & $t$ & $w$              & {\sc Sign}& {\sc MLSS-S} & {\sc Verify} & {\sc MLSS-V.1} & {\sc MLSS-V.2}\\ 
\hline 16,384 & 16 & 16 & 16                & 1.5238     & 1.6200             & 0.1368         & 0.1398 & 0.2330\\
\hline 65,536 & 64 & 64 & 64                & 1.7951     & 2.1796             & 0.4081         & 0.4196 & 0.7926\\
\hline 262,144 & 256 & 176 & 2816          & 2.8804     & 4.8561            & 1.4934         & 1.5246 & 3.4691\\
\hline 1,048,576 & 1024 & 352 & 11264       & 7.2215     & 15.0616           & 5.8345         & 5.8968 & 13.6746\\
\hline
\end{tabular}

\end{table*}

As we can see in Tables 1 to 8, the running time of {\sc MLSS-Sign} is on average the double of a traditional {\sc Sign} for the biggest documents. Otherwise, for the smaller documents, the running time can be similar to {\sc Sign} (see columns {\sc Sign} and {\sc MLSS-S}). If the document is not modified or if the verifier is not interested in locating the modifications, the cost of {\sc MLSS-Verify} is basically the same as a regular RSA {\sc Verify} for all sizes of documents or values of $d$ (see columns {\sc Verify} and {\sc MLSS-V.1} for different values of $d$). Finally, the cost to locate $d$ modifications with {\sc MLSS-V.2} remain the double of {\sc Verify} even when the size of the document and the value of $d$ increase for block sizes of 8192 bytes. For documents divided into more blocks (blocks of size 1024 bytes), the cost of {\sc MLSS-V.2} can be a little more than the double, specially in the cases where the values of $b$ and $d$ are large (see columns {\sc Verify} and {\sc MLSS-V.2} for different values of $d$).

In conclusion, for the values considered we observe a moderate increase
in the running time of {\sc MLSS-Sign} and {\sc MLSS-Verify} with respect to {\sc Sign} and {\sc Verify},
and this increase is highest when $n$ and $d$ are larger.
For the documents with $n=1024$ the increase was by a factor between 2 and 2.5
when $d$ varies from 1 to 10.
We also remark that while this increase of time is always incurred 
for {\sc MLSS-Sign}, it is not incurred in {\sc MLSS-Verify} if no modifications
occurred or if the verifier does not wish to locate modifications ($lc$=false).


\subsection{Division in blocks and block sizes}\label{sec:blocks}

An issue that needs to be considered is the scheme to divide the document into $n$ blocks. Sender and receiver must use the same block organization, and we require that this organization must be preserved.
For instance, dividing a text into blocks of the same size makes it hard to support modifications, as one bit inserted into block 1 would prevent us to keep track of where the other blocks are. Therefore, information on block structure must be included with the document (e.g.~a description header) and legitimate modifications should be done using a system that is ``block aware".

We suggest two possible solutions to the problem of dividing the document into blocks. A first solution 
is to use special delimiters to separate blocks (e.g.~tags on an XML document, or reserved characters on a text) or a description header that indicates where each block starts. A second solution is to use the own data organization to separate blocks (e.g.~the records of a database). We could also use the semantics of the data to separate blocks for a specific application (e.g.~sections of a document).

A second issue to be considered is the block size, which depends on the
application
needs and computational capabilities.
The extreme values of block size may not be suitable. Block size equal to
$b$ ($n=1$)
means no ability to locate modifications, while very small block size
makes $n$ and
$t$ too large, rendering the scheme inefficient. We can observe the effect
of number of blocks on running time for different values of $d$ in Tables
1 to 8.

\section{Conclusion}\label{sec:conclusion}
In this paper, we propose a general Modification Location Signature Scheme
(MLSS), where we decouple the verification of the signature from the
possible modifications in the
document. Our method is the first to address the issue of locating
modified blocks.
This is accomplished with some additional costs in hash computations
but no additional costs in the cryptographic functions involved ({\sc Sign} and {\sc Verify}).

\section*{Acknowledgements}
Thaís Bardini Idalino was supported by a CAPES-Brazil scholarship; Lucia Moura and Daniel Panario were partially supported by an NSERC-Canada discovery grant; Daniel Panario was partially funded by a grant of the PVE program of CAPES-Brazil. We would like to thank the referees for several suggestions that improved
this paper.

\end{document}